\documentclass{article}

\newif\ifgreyprint
\greyprintfalse
\ifgreyprint\else\fi
 
\usepackage{latexsym,xspace,calc,amsthm}
\usepackage{amsmath,amssymb} %
\usepackage{nicefrac} %

\usepackage{lmodern}
\usepackage{microtype}
\usepackage{tgtermes}
\usepackage{fix-cm}

\usepackage{bm} %

\usepackage{cancel}

\usepackage{amssymb,stmaryrd,amsmath}
\usepackage{mdwlist} %
\usepackage{float}   %
\usepackage{centernot} %

\PassOptionsToPackage{hyphens}{url}  %
\ifgreyprint
\usepackage[allcolors=black,colorlinks]{hyperref}
\else
\usepackage[colorlinks]{hyperref}
\fi
\usepackage{relsize}

\usepackage{breakurl}  %

\newtheoremstyle{jamiestyle}%
  {4pt}%
  {0pt}%
  {\it}%
  {0pt}%
  {\sc}%
  {.}%
  { }%
  {}%
\theoremstyle{jamiestyle}
\newtheorem{thrm}{Theorem}[section]

\newtheorem{lemm}[thrm]{Lemma}

\newtheoremstyle{jamienfstyle}%
  {4pt}%
  {0pt}%
  {\normalfont}%
  {0pt}%
  {\sc}%
  {.}%
  { }%
  {}%
\theoremstyle{jamienfstyle}
\newtheorem{nttn}[thrm]{Notation}
\newtheorem{defn}[thrm]{Definition}
\newtheorem{xmpl}[thrm]{Example}

\newcommand\jamiesection[1]{\section{#1}}
\newcommand\jamiesubsection[1]{\subsection{#1}}

\newcommand\tvT{{\mathbf t}}
\newcommand\tvF{{\mathbf f}}

\newcommand\opens{\tf{Open}}

\newcommand{\dotarrow}{%
   \mathrel{\ooalign{\hss\raise.85ex\hbox{\scalebox{1.25}{\normalfont .}}%
   \kern0.35ex\hss\cr$\rightarrow$}}}

\usepackage{letltxmacro}%
\newcounter{fnmarkcntr}\newcounter{fntextcntr}
\makeatletter
\renewcommand{\footnotemark}{%
   \@ifnextchar[\@xfootnotemark
     {\stepcounter{fnmarkcntr}%
      \refstepcounter{footnote}\label{footnotemark\thefnmarkcntr}%
      \protected@xdef\@thefnmark{\thefootnote}%
      \@footnotemark}}
\makeatother
\LetLtxMacro{\oldfootnotetext}{\footnotetext}%
\renewcommand{\footnotetext}[1]{%
  \stepcounter{fntextcntr}%
  \oldfootnotetext[\ref{footnotemark\thefntextcntr}]{#1}%
}

\newcommand\onlineref[2]{\url{#1} (permalink: \url{#2})}
\newcommand\footnoteref[2]{\footnote{See \onlineref{#1}{#2}.}}

\newcommand\intertwinedwith{\mathrel{\between}}

\makeatletter
\newcommand\@deffont[2][]{{\bfseries #2}\index{#1}}
\newcommand\deffont{\@dblarg\@deffont}
\makeatother
\newcommand\powerset{\f{pow}}

\newcommand\f[1]{\mathit{#1}}
\newcommand\tf[1]{\mathsf{#1}}
\newcommand\ns[1]{\bm{\mathsf{#1}}}

\newcommand\limp{\Longrightarrow}

\DeclareMathSymbol{\shortminus}{\mathbin}{AMSa}{"39}
\newcommand\minus{{\shortminus}}
\newcommand\plus{{+}}
\newcommand\Forall[1]{\forall #1.}

\newcommand\cent{\vdash}

\makeatletter
\DeclareRobustCommand{\barcent}{\mathbin{\mathpalette\barcent@@\relax}}
\newcommand{\barcent@@}[2]{%
  \vbox{\offinterlineskip
    \sbox\z@{$\m@th#1\cent$}%
    \ialign{%
      \hfil##\hfil\cr
      $\m@th#1{}_{\minus}\kern-\scriptspace$\cr
      \noalign{\kern-.3\ht\z@}
      \box\z@\cr
    }%
  }%
}
\makeatother

\makeatletter
\def\pmb@#1#2{\setbox8\hbox{$\m@th#1{#2}$}%
  \setboxz@h{$\m@th#1\mkern-.1mu$}\pmbraise@\wdz@
  \binrel@{#2}%
  \dimen@-\wd8 %
  \binrel@@{%
    \mkern-.1mu\copy8 %
    \kern\dimen@\mkern-.2mu\copy8 %
    \kern\dimen@\mkern-.3mu\copy8 %
    \kern\dimen@\mkern-.4mu\copy8 %
    \kern\dimen@\mkern.1mu\copy8 %
    \kern\dimen@\mkern.2mu\copy8 %
    \kern\dimen@\mkern.3mu\copy8 %
    \kern\dimen@\mkern.0mu\raise\pmbraise@\copy8 %
    \kern\dimen@\mkern.4mu\box8 %
           }%
}
\makeatother

\makeatletter
\newcommand{\circlearrow}{}%
\DeclareRobustCommand{\circlearrow}{%
  \mathrel{\vphantom{\shortrightarrow}\mathpalette\circle@arrow\relax}%
}
\newcommand{\circle@arrow}[2]{%
  \m@th
  \ooalign{%
    \hidewidth$#1\circ\mkern1mu$\hidewidth\cr
    $#1\longrightarrow$\cr}%
}
\makeatother

\makeatletter
\newcommand*\bigcdot{\mathpalette\bigcdot@{.5}}
\newcommand*\bigcdot@[2]{\mathbin{\vcenter{\hbox{\scalebox{#2}{$\m@th#1\bullet$}}}}}
\makeatother

\usepackage{datetime}
\yyyymmdddate

\title{Decentralised collaborative action: cryptoeconomics in space} 
\author{Murdoch J. Gabbay}

\begin{document}

\renewcommand{\today}{}
\maketitle

\pagenumbering{arabic}
\setcounter{page}{1}

\subsection*{Abstract}
Blockchains and peer-to-peer systems are part of a trend towards computer systems based on \emph{decentralised collaborative action}, by which we mean that they \emph{1)} run across many participants, \emph{2)} without central control, and \emph{3)} are such that qualities~1 and~2 are essential to the system's intended use cases.

We propose a notion of topological space, which we call a \emph{semitopology}, to help us mathematically model such systems.
We treat participants as \emph{points} in a space, which are organised into \emph{actionable coalitions}.
An actionable coalition is any set of participants who collectively have the resources to collaborate (if they choose) to progress according to the system's rules, independently of the rest of the system.

Mathematicians will recognise semitopology as a generalisation of the notion of point-set topology, where actionable coalitions correspond to open sets.

It turns out that much useful information about the system can be obtained \emph{just} by viewing it as a semitopology and studying its actionable coalitions. 
For example: we will prove a mathematical sense in which if every actionable coalition of some point $p$ has nonempty intersection with every actionable coalition of another point $q$ --- note that this is the negation of the famous Hausdorff separation property from topology --- then $p$ and $q$ must remain in agreement.
Remarkably, since this observation depends only on the semitopological structure, it holds for any possible concrete algorithm.

This matters because remaining in agreement is a key correctness property in many distributed systems.
For example in blockchain, participants disagreeing is called \emph{forking}, and blockchain designers try hard to avoid it. 

We provide an accessible introduction to: the technical context of decentralised systems; why we build them and find them useful; how they motivate the theory of semitopological spaces; and we sketch some basic theorems and applications of the resulting mathematics.

\newpage
\jamiesection{What is `cryptoeconomics'?}
\label{sect.intro}

Let us begin by proposing a definition:
\begin{defn}
\label{defn.c}
\deffont{Cryptoeconomics} is:
\begin{quote}
the study of socioeconomic systems enabled by modern decentralised computer systems. 
\end{quote}
\end{defn}
Cryptoeconomics is closely related to blockchains, because a blockchain is a decentralised database, and databases store \emph{state} (i.e. data): thus, blockchains make cryptoeconomics possible, and cryptoeconomic outcomes --- real or envisaged --- make blockchains useful. 

There is substantial overlap between cryptoeconomics and many other fields that study aspects of decentralised and distributed systems, including:
\begin{itemize*}
\item
\emph{economics} (the study of value and incentives), 
\item
\emph{game theory} (the study of optimal outcomes for players with choices), 
\item
\emph{social choice theory} (how to synthesise collective decisions from individual votes), 
\item
\emph{blockchains} (decentralised databases), 
\item
\emph{smart contracts} (programs that operate on blockchains), and
\item
\emph{law} (the interpretation of facts into socially-binding meaning).
\end{itemize*}
We see from this list that in practice, cryptoeconomics touches on nearly everything.

The description of cryptoeconomics in Definition~\ref{defn.c} uses two adjectives: `digital' and `decentralised'.
Digitisation is important because it enables decentralisation on a historically unprecedented scale, but it is the decentralisation that is most important to giving cryptoeconomics its particular character.

Concerning digitisation, 
this has been ongoing since (at least) the digital \emph{mainframe computers} of the 1960s ---
powerful central computers that acted then, and still act today, as oracles to enable organisations that can afford them to deliver services more efficiently.
Note however that this is invisible to the end user, in the sense that the user just sees a bigger, better, faster, and possibly cheaper service.
This is a substantive technological advance, and great for efficiency and profits, but it does not necessarily lead to any qualitative structural change in the economics of how the value is created.

We can date the seeds of \emph{decentralised} digitisation to the mid-1970s, when pocket digital calculators took over from slide rules, and desktop computers became available.
Putting digital computation in people's pockets and on people's desks started a cascade of innovations that, along with the internet, has brought us modern miracles like mobile banking and streaming video.
Yet even so, a user doing mobile banking on a mobile phone or watching a video on a streaming service has something fundamental in common with the technician querying a 1960s mainframe computer using a paper card with holes punched in it: the back-end system is still centrally controlled.
That mainframe computer system is still around, albeit in an immensely more sophisticated form.\footnote{In practice, the modern `mainframes' that (for example) serve banking, social media, or streaming services, are usually distributed clusters of servers.  But being distributed is not the same as being decentralised: see Definition~\ref{defn.parallel}(1).}

What makes the new breeds of modern computing systems uniquely different is that they are \emph{radically decentralised} and heterogeneous, such that \emph{if they were centrally controlled then they would not even make sense}.
This is the start of cryptoeconomics as intended in Definition~\ref{defn.c}, and it is new.

It began with Napster (a peer-to-peer filesharing system) in June 1999, which demonstrated how music media could be disseminated independently of (centralised) media companies~\cite{David2016}.\footnote{The article itself illustrates some of the compromises involved in creating and disseminating knowledge.
The article's author is an Associate Professor at a UK University.
The publisher's version is behind a paywall (since the publisher makes money from publishing); an author's preprint is made freely online by his University (which makes money by employing the author to educate students); and the author most likely wrote the article in-between teaching obligations, for zero marginal cost to his employer (i.e. `for free'). 

To be fair: the publisher's version looks nicer than the author's preprint, and the publisher's website makes the work easy to find; the University survives and its students get taught by a well-informed professor; the author's preprint is accessible to any reader who can dig it out; and the author enjoyed writing the article.
In this sense, all parties --- the author, the publisher, the students, the university, the article's readership, and science itself --- benefit from the shambling compromise that is academic publishing, though they might also all complain about the division of that benefit.}
Then Bitcoin (January 2009) and Ethereum (July 2015) changed everything by showing that money and even contracts could be mediated (albeit imperfectly) independently of banks.

These systems --- of which Napster, Bitcoin, and Ethereum are perhaps the best known, but certainly not the only examples --- would be meaningless and make no sense if they were implemented in a centralised manner, in much the same way that a pocket calculator or desktop computer would make no sense if it were a mainframe.
The decentralisation is not a quality of the thing: it is the \emph{point} of that thing.

This tendency towards radical decentralisation is often discussed in ideological terms, but compelling technical forces also exist to push technology in this direction.
Being decentralised gives desirable properties, including: scalability, redundancy, reliability, and resilience.
Being decentralised offers unique opportunities for a network of participants to \emph{act collaboratively} to create value by achieving their goals.%
We can sum this up as follows:
$$
\text{cryptoeconomics} 
\quad = \quad
\text{value} + \text{decentralised collaborative action} .
$$
We are only beginning to get to grips with the implications of this equation.

\jamiesection{What is `decentralised collaborative action'?}
\label{sect.what.is}

\jamiesubsection{The definition}

\emph{Decentralised collaborative action} is a feature of \emph{decentralised permissionless heterogeneous computing systems}.
Let's unpack the jargon:
\begin{defn}
\label{defn.parallel}
\leavevmode
\begin{enumerate*}
\item
A system is \emph{decentralised} when it is \emph{distributed} (meaning that it is composed of several distinct parts), and furthermore the system as a whole is not centrally controlled.

Most blockchain systems and peer-to-peer networks are decentralised in this sense. 
The internet is also (mostly) decentralised, at least in principle.\footnote{The internet was designed to be an information network that would be resilient to nuclear attack.  It did this by being `centrifugal'; emphasising node-to-node actions instead of centre-to-centre actions.  See~\cite{ryan:hisidf}, summarised by Ars Technica~\cite{ars-technica:howabg}.

Note that the boundary between `distributed' and `decentralised' can be fluid.  For example, should we consider a system to be decentralised when its parts can act independently most of the time, but every so often they check in with a central controller?  This depends on what aspects of the system we care about; e.g. its short-term or long-term behaviour. There is room for a nice discussion here, but it will not be in this particular article.} 
\item
A system is \emph{permissionless} (or \emph{unpermissioned}) when participants can leave and join the system at any time.

Nature is naturally permissionless (living things do not need permission to be born or die).
National voting systems \emph{are} permissioned (because citizens require government certification to be allowed to vote).
Peer-to-peer systems (including filesharing and blockchain protocols) are often, though not always, unpermissioned. 
\item
A system is \emph{heterogeneous} when participants may legitimately be following different rules.\footnote{By `different rules' we include the situation where an algorithm (such as a consensus algorithm) is agreed between participants but a critical parameter may vary substantively across them, e.g. imagine a blockchain in which some participants require a ${>}2/3$ majority to act, and others require just a ${>}1/2$ majority. By `legitimate' we exclude the case of a hostile participant.}

Ethereum and Tezos are decentralised and permissionless, but they are not heterogeneous in the sense we intend.
If you are running a Tezos or Ethereum node, then you are not forced to follow the rules, but if you do not then by definition you are not acting as a legitimate Tezos or Ethereum node. 

In contrast, consider the combination of Tezos and Ethereum as a single system connected by a \emph{blockchain bridge}.\footnoteref{https://ethereum.org/en/bridges/}{https://web.archive.org/web/20240324090911/https://ethereum.org/en/bridges/}
This is heterogeneous, because Tezos nodes and Ethereum nodes have different rules and different consensus mechanisms.
A Tezos node is not a bad node just because it is not following the rules of Ethereum, and vice-versa, but because of the blockchain bridge, they can be considered to be operating within a single (heterogeneous) combined system.
\end{enumerate*}
\end{defn}
There are many flavours of decentralised system, but in the most general case we have a decentralised heterogeneous permissionless system that consists of \emph{some} participants communicating to do \emph{something}, with no \emph{a priori} restrictions on who, what, or how.

This scenario --- with its weak well-behavedness assumptions that do not even assume all participants share a common ruleset --- might seem a terrible idea which we should not allow, because it admits 
crazy networks with bad behaviour.
But here the generality is a feature, not a bug:
\begin{enumerate*}
\item
Mathematically speaking, it can be \emph{useful} to admit general models, including both good and bad ones, so that we can formalise their good and bad behaviour\footnote{\dots which will vary by application; e.g. sometimes all participants should play by the same rules, but in the case of a blockchain bridge we specifically want to \emph{admit} different rules.} and express conditions to include or exclude it.
\item
Surprisingly, it will turn out that there is still a lot that we can say even about the general case, and we shall see that much useful structure will emerge, even from very weak assumptions. 
\end{enumerate*} 
So granted that the generality of decentralised collaborative action is a feature, not a bug; but how should we approach this mathematical generality?

The key is to look at how groups of participants can \emph{progress (i.e. update) their local state}.
To see this we need a little more discussion.

In a decentralised system, a participant must store local state --- if there were a global source of truth for state then whoever controls that truth would \emph{de facto} control the system --- 
and communicate with other participants to decide on how their local states evolve.
There must be \emph{some} rules about how this state should be updated, even if these rules may differ across participants in the system, and even though the rules may not always be followed. 
It turns out that one common feature of decentralised heterogeneous permissionless systems is a notion of what in~\cite{gabbay:semdca,gabbay:semtad} is called an \emph{actionable coalition},
by which we mean 
\begin{quote}
\emph{a set of participants who are legally entitled (but not obliged) to collaborate to progress and update their local state (possibly but not necessarily in identical ways).}
\end{quote}

\jamiesubsection{Some informal examples from real life}

We will consider some examples of actionable coalitions, taken from real life:
\begin{enumerate*}
\item
Ethereum.

Ethereum's consensus protocol is proof-of-stake, so an actionable coalition on Ethereum is any group of participants who hold a majority stake of tokens (this is a bit of a simplification, but it will do).
\item
Ethereum and Tezos with a blockchain bridge between them.

Tezos's consensus protocol is also proof-of-stake.
An actionable coalition in this system is either an actionable coalition of Ethereum, or one of Tezos, \emph{or} the sets union of an actionable coalition from each, along with the bridging node (again, a simplification, but it will do).\footnote{Typically, participants can update their state if they held a majority of the stake at some time in the past (e.g. two weeks ago) --- the idea being that all participants have reached agreement on, and learned, the state of the network two weeks in the past, so this can be treated as immutable common knowledge without undermining the decentralised nature of the system in the present~\cite[Subsection~3.2.1, final paragraph]{tezos:whitepaper}.} 
We return to this in Example~\ref{xmpl.ET}. 
\item
A Tango dance evening where leaders will only dance with followers and vice-versa.\footnote{Many dancers can both lead and follow, including this author, but for the sake of the mathematics we will simplify.}

An actionable coalition is any set containing equal numbers of leads and followers.
\item
A set of people wishing to lift a heavy rock.

An actionable coalition is any subset of these people who lifting together have enough strength to do so. 
\end{enumerate*}
If an actionable coalition can communicate to agree on a set of local state updates, e.g. if the Tango lead leads a move and the Tango follower chooses to follow it, then the participants in this coalition are entitled to update their local states accordingly.
Note that local state updates need not be literally identical across participants; they just need to be mutually agreed upon and then actioned.

Some important notes:
\begin{enumerate*}
\item
The actionable coalition can progress \emph{without} consulting the rest of the system.
\item
Being in an actionable coalition does not imply control.
This set describes a potential legal collaboration, but participants can choose what actionable coalition to work with, if any, and they can also choose not to follow the rules.\footnote{If you put your elbow into your dance partner's eye, or simply deliver a poor lead or a poor follow, then the other dancer might stop dancing with you or turn you down if you ask for another dance.
But neither of you are \emph{compelled} to dance with one another, and if you do, you are not \emph{compelled} to dance well.}
\item 
If $O$ is an actionable coalition for some participant $p$, and $p'$ is another participant in $O$, then $O$ is also an actionable coalition for $p'$.
Note that this makes actionable coalitions look a bit like open sets in a topology.
\end{enumerate*}
So we can now introduce our first mathematical abstraction: we identify participants as \emph{points}, and we let \emph{open sets} be \emph{actionable coalitions}.
An actionable coalition is a \emph{coalition of participants with the capacity to act}.
They are not obliged to act, and if they do act then their action need not be identical across all participants, but the potential exists for this set to collaborate to progress their states.
\begin{enumerate*}
\item
With reference to our couples dance example:
an example of an actionable coalition that is not minimal is a set containing two leads and two followers.
There are two ways for the participants to pair off to collaborate (i.e. dance).
\item
With reference to our bridged blockchain example: an example of a set that contains an actionable coalition but is not one itself is an actionable coalition from Ethereum, along with the bridging node.
The Ethereum coalition on its own is actionable, but the bridging node cannot take any action without also collaborating with an actionable coalition from Bitcoin.
\end{enumerate*} 

To get a flavour of our mathematical results, consider a fundamental problem in any decentralised system: \emph{consensus}; i.e. the problem of ensuring that participants remain in agreement, for some suitable sense of `agree'.
To take a simple example from blockchain: if we reach a situation where half of the nodes say that we have paid for a service, and the other half say that we have not --- then \emph{everyone} has a problem, because the system has become incoherent and it is not clear how the system as a whole can restore coherence and progress.\footnote{coherent (adj.) 1550s, ``harmonious;'' 1570s, ``sticking together,'' also ``connected, consistent'' (of speech, thought, etc.), from French cohérent (16c.), from Latin cohaerentem (nominative cohaerens), present participle of cohaerere ``cohere,'' from assimilated form of com ``together'' (see co-) + haerere ``to adhere, stick'' (etymologyonline: \url{https://www.etymonline.com/word/coherent}).}
This phenomenon is called \emph{forking}, and blockchain designers really want to avoid it!

We will call our mathematical abstraction of agreement, \emph{antiseparation}.
In a little more detail, antiseparation properties are coherence properties that are guaranteed to hold of a decentralised system 
\emph{just} by analysing the structure of its actionable coalitions.
It turns out that we can get surprisingly detailed information about agreement/antiseparation properties, even working from quite weak and abstract mathematical assumptions on the actionable coalitions.

We emphasise this point: sometimes we can predict important macro properties of a system's behaviour without knowing anything about its specifics, so long as we have certain good properties on its actionable coalitions.

\jamiesubsection{Two formal mathematical examples}

\begin{defn}
\label{defn.binary.consensus}
Call \deffont{binary consensus} the problem of getting participants in a distributed system to announce a single value $\tvT$ or $\tvF$.
This is a simplest possible consensus problem, but note that by running multiple rounds of binary consensus we can get participants to announce \emph{bitstrings} (finite sequences of values), and arbitrary data can be serialised to bitstrings, so this consensus problem --- simple as it is --- is also complete for all data in a suitable sense.
\end{defn} 

\begin{xmpl}[A simple majority system]
\label{xmpl.majority}
We consider a simple situation where participants are trying to solve binary consensus.
Continuing the theme of simplicity, assume some finite nonempty set of participants $\ns M$ and let their actionable coalitions be just any set of participants that forms a majority (so it contains strictly more than half of $\ns M$ the set of all participants).
Now suppose that the participants in some actionable coalition $O\subseteq\ns M$ ($O\subseteq\ns M$ means that $O$ is a set of elements from $\ns M$) have communicated and have progressed to agree on the value $\tvT$. 
Because they form an actionable coalition, they are entitled to act and to announce $\tvT$, and so they do.
They have now all committed to this state update and they cannot change their minds.

So: can this system fork?
Consider some participant $p\not\in O$ ($p\not\in O$ means that $p$ is a participant that is not in $O$).
If $p$ wants to progress to decide on some value, that value must be $\tvT$.
This is because all of its actionable coalitions intersect with $O$, and so they contain at least one participant that has already committed to $\tvT$ and cannot change its mind.

This does not mean that $p$ has to agree on $\tvT$; it could choose not to agree with anything and not progress (i.e. not update its local state with any value), or it could break the rules.
But, by definition if $p$ does want to made a decision legally, then the decision has been made and it must eventually go along with the majority.
Thus, we have proved that any progress that is made by one participant within the rules (\dots must be shared with some actionable coalition of that participant, and since all such coalitions intersect it \dots) must eventually be followed any other participant that also progresses within the rules.
Thus forking is impossible.
\end{xmpl}

The reader may already be familiar with Example~\ref{xmpl.majority}, but note that this antiseparation property comes simply \emph{from the structure of the actionable coalitions}.
There is no need to consider the protocol, or even how values are interpreted.
It turns out that antiseparation-style behaviour is common, and arises even if we do not require actionable coalitions that are simple majorities.
For example:
\begin{xmpl}
\label{xmpl.Z}
Let participants be $\mathbb Z=\{0,1,\minus 1,2,\minus 2,\dots\}$ and let actionable coalitions be generated by sets of three consecutive numbers starting at an even number $\{2i,2i\plus 1,2i\plus 2\}$ --- so for example $\{0,1,2\}$ and $\{2,3,4\}$ are actionable coalitions, but not $\{2\}$ or $\{1,2,3\}$ --- and suppose again that we are trying to agree on $\tvT$ or $\tvF$.

Note that unlike for Example~\ref{xmpl.majority}, actionable coalitions need not intersect.
Yet, once one triplet of participants commits to $\tvT$, the rest of the system is obliged to eventually agree, if all participants play by the rules.
Now this example system is not necessarily particularly safe or desirable in practice, because we can imagine that $\{0,1,2\}$ agree on $\tvT$, and $\{4,5,6\}$ acting independently but in good faith agree on $\tvF$, and then $3$ cannot legally progress, because within $\{2,3,4\}$, $2$ has announced $\tvT$ and $4$ has announced $\tvF$ and $3$ cannot agree with both.
But, we know that \emph{if} all participants do legally progress, then they announce the same value.
So this example illustrates how antiseparation can arise even when actionable coalitions are rather small.
\end{xmpl}

The two examples above are quite different.
In one, all actionable coalitions intersect, and in the other they mostly do not.
This suggests that a `general mathematics of (anti)separation' is possible, based on the study of actionable coalitions.
In a nutshell, that mathematical story is what we will develop.

\jamiesection{What is a semitopology?}

\jamiesubsection{The definition}

So at a high level, what do we have?
Points are synonymous with participants, and:
\begin{enumerate*}
\item
There is a notion of an \emph{actionable coalition} (or just: \emph{open set}).
This is a set $O\subseteq\ns P$ of participants with the capability, though not the obligation, to act collaboratively to advance (= update / transition) the local state of the elements in $O$, possibly but not necessarily in the same way for every $p\in O$. 
\item
The empty set $\varnothing$ (containing no points) is trivially an actionable coalition.
Also we assume that $\ns P$ (containing all the points) is actionable, effectively assuming that every point is a member of at least one actionable coalition.
\item
A sets union of actionable coalitions, is an actionable coalition.
\end{enumerate*}
This leads us to the definition of a semitopology.

\begin{defn}
\label{defn.semitopology}
Suppose $\ns P$ is a set.
Write $\powerset(\ns P)$ for the powerset of $\ns P$ (the set of subsets of $\ns P$).
Then a \deffont{semitopological space}, or \deffont{semitopology} for short, consists of a pair $(\ns P, \opens(\ns P))$ of 
\begin{itemize*}
\item
a (possibly empty) set $\ns P$ of \deffont{points}, and 
\item
a set $\opens(\ns P)\subseteq\powerset(\ns P)$ of \deffont[open sets $\opens$]{open sets}, 
\end{itemize*}
such that:
\begin{enumerate*}
\item\label{semitopology.empty.and.universe}
$\varnothing\in\opens(\ns P)$ and $\ns P\in\opens(\ns P)$.

In words: the empty set of points, and the set of all points, are both open sets.
\item\label{semitopology.unions}
If $X\subseteq\opens(\ns P)$ then $\bigcup X\in\opens(\ns P)$.\footnote{There is a little overlap between this clause and the first one: if $X=\varnothing$ then by convention $\bigcup X=\varnothing$.  Thus, $\varnothing\in\opens(\ns P)$ follows from both clause~1 and clause~2.  If desired, the reader can just remove the condition $\varnothing\in\opens(\ns P)$ from clause~1, and no harm would come of it.} 

In words: a sets union of open sets, is an open set.
\end{enumerate*}
We may write $\opens(\ns P)$ just as $\opens$, if $\ns P$ is irrelevant or understood. %
\end{defn}

We recognise a semitopology as being like a \emph{topology}~\cite{engelking:gent,willard:gent}, but without the condition that the intersection of two open sets necessarily be an open set.
This reflects the fact that the intersection of two actionable coalitions need not itself be an actionable coalition.

\begin{nttn}
Suppose $X$ and $X'$ are sets.
Then write $X\between X'$ when $X$ and $X'$ are not disjoint; i.e. they have a nonempty sets intersection: $X\cap X'\neq\varnothing$. 
\end{nttn}

We can now state a key definition:
\begin{defn}
Suppose $(\ns P,\opens)$ is a semitopology and $p,p'\in\ns P$.
Then:
\begin{enumerate*}
\item
Write $p\intertwinedwith p'$ and call $p$ and $p'$ \deffont{intertwined} when
$$
\Forall{O,O'\in\opens}p\in O\land p'\in O' \limp O\between O' 
$$
In words: $p$ and $p'$ are intertwined when they have no pair of disjoint open neighbourhoods.
\item
Write $p\intertwinedwith^* p'$ and call $p$ and $p'$ \deffont{transitively intertwined} when they are related by the transitive closure of $\intertwinedwith$; thus there exists some (possibly zero length) chain $p_0,p_1,\dots,p_{n\minus 1},p_n\in\ns P$ such that 
$$
p=p_0\intertwinedwith p_1\dots p_{n\minus 1}\intertwinedwith p_n=p'.
$$
\end{enumerate*}
\end{defn}

This recalls the \emph{Hausdorff separation} property, typical in topology, that any two distinct points should have a disjoint pair of open neighbourhoods.
As we shall see from Theorem~\ref{thrm.1}, 
for the study of consensus we are particularly interested in semitopologies with antiseparation properties, of which being intertwined is a canonical such property; and it is the precise negation of being Hausdorff separated.

\jamiesubsection{Two mathematical results}
\label{subsect.two.mathematical.results}

Recall the notion of \emph{binary consensus} from Definition~\ref{defn.binary.consensus}.
Now that we have built some mathematical machinery, we can represent binary consensus as the problem of defining a function $f:\ns P\to\{\tvT,\tvF\}$.
We will call such a function a \deffont{value assignment}.

If the value assignment is constant (so it maps all points to just one value) then this represents system-wide consensus across all of $\ns P$.

Call a value assignment \deffont{continuous} at $p\in\ns P$ when there exists an open neighbourhood $p\in O\in\opens$ such that $\Forall{p'\in O}f(p)=f(p')$.
The reader can check that this coincides with the usual notion of topological continuity, if we give $\{\tvT,\tvF\}$ the discrete topology (so $\{\tvT\}$ and $\{\tvF\}$ are open sets); a proof is in~\cite[Lemma~2.2.4, page~22]{gabbay:semdca}.
It also coincides with our intuition that if $p$ declares some value, then it must do so in collaboration with an actionable coalition of other participants.  
Thus we can write
\begin{quote}
consensus = continuity,
\end{quote}
and we can prove:
\begin{thrm}
\label{thrm.1}
Suppose $(\ns P,\opens)$ is a semitopology and $p,p'\in\ns P$, and suppose $f:\ns P\to\{\tvT,\tvF\}$ is a value assignment.
Then:
\begin{enumerate*}
\item
If $p$ and $p'$ are intertwined and $f$ is continuous at $p$ and $p'$, then $f(p)=f(p')$.
\item
If $p$ and $p'$ are transitively intertwined and $f$ is continuous at all points, then $f(p)=f(p')$.
\end{enumerate*}
\end{thrm}
\begin{proof}
For part~1, suppose $f$ is continuous at $p$ and $p'$, and suppose $p\intertwinedwith p'$.
Then by assumption there exist open neighbourhoods $p\in O\in\opens$ and $p'\in O'\in\opens$ such that $f$ is constant on $O$ and on $O'$.
Since $O$ and $O'$ intersect, $f(p)=f(p')$.

For part~2, suppose $f$ is continuous at every point, and suppose $p\intertwinedwith^* p'$.
By assumption there exists a chain of intertwinedness relations $p=p_0\intertwinedwith p_1\dots p_{n\minus 1}\intertwinedwith p_n=p'$, and also by assumption $f$ is continuous at each of these points.
By part~1 of this result $f(p)=f(p_0)=f(p_1)=\dots=f(p_{n\minus 1})=f(p_n)=f(p')$.
\end{proof}

Simple as Theorem~\ref{thrm.1} is, it explains the consensus behaviour we observed of Examples~\ref{xmpl.majority} and~\ref{xmpl.Z}, via the following easy Lemma:
\begin{lemm}
All points in the semitopology in Example~\ref{xmpl.majority} are intertwined.
All points in the semitopology in Example~\ref{xmpl.Z} are transitively intertwined.
\end{lemm}
\begin{proof}
Left to the reader to check.\footnote{Hint: {\tiny check that $2i\plus 1$ is intertwined with $2i$ and $2i\plus 2$.}} 
\end{proof}

\begin{xmpl}
\label{xmpl.ET}
Consider our previous example of Ethereum and Tezos, connected by a bridging node.
What would this look like in abstract semitopological terms?

Assume two semitopologies $(\ns E,\opens(\ns E))$ and $(\ns T,\opens(\ns T))$ such that (for simplicity) $\ns E\cap\ns T=\varnothing$.
Assume some other point $r\notin\ns E\cup\ns T$, which we will call the \deffont{bridging node}.
Then define a semitopology $(\ns B,\opens(\ns B))$ by:
\begin{itemize*}
\item
$\ns B=\ns E\cup\{r\}\cup\ns T$.
\item
$\opens(\ns B)$ is the closure under arbitrary unions of 
$$
\opens(\ns E)\cup\opens(\ns T)\cup \{O\cup\{r\}\cup O' \mid O\in\opens(\ns E),O'\in\opens(\ns T)\} .
$$
\end{itemize*}
Intuitively, a quorum in the combined system is either a quorum from $\ns E$, or one from $\ns T$, or it is a pair of quorums along with the bridging node $r$.
The key point about this structure is that the bridging node can only make progress if it is in consensus with some quorum from $\ns E$ \emph{and} at the same time some quorum from $\ns T$.
This is just what we would expect a bridging node to do.

Now suppose that $(\ns E,\opens(\ns E))$ and $(\ns T,\opens(\ns T))$ are intertwined --- which is what we would hope, since this indicates a pair of blockchains that will not fork.
Then $(\ns B,\opens(\ns B))$ is \emph{transitively} intertwined, via the bridging node.
\end{xmpl}

\jamiesection{Conclusions and open questions}

\jamiesubsection{Overview}

We have discussed \emph{semitopology}, a generalisation of point-set topology that removes the restriction that intersections of open sets need necessarily be open.
The intuition is that points represent participants in a decentralised system, and open sets represent collections of participants that collectively have the authority to collaborate to update their local state; we call this an \emph{actionable coalition}.

Examples of actionable coalition include: majority stakes in proof-of-stake blockchains; communicating peers in peer-to-peer networks; and even pedestrians working together to not bump into one another in the street.
Where actionable coalitions exist, they have in common that: collaborations are local (updating the states of the participants in the coalition, but not immediately those of the whole system); collaborations are voluntary (up to and including breaking rules); participants may be heterogeneous in their computing power or in their goals (not all pedestrians want to go to the same place); participants can choose with whom to collaborate; and they are not assumed subject to permission or synchronisation by a central authority.

These decentralised systems can be very complex, and without a centralised authority to control behaviour, it is not immediately obvious why they should display order.
Semitopologies are a topology-flavoured mathematics that goes some way to explaining how and in what circumstances they can behave well. 
Semitopology is also interesting in and of itself, having a rich and interesting theory --- one which quickly deviates from standard accounts on topological spaces, because the most interesting semitopologies are rather ill-behaved from the usual viewpoint, as their antiseparation properties mean that they are not Hausdorff.
Various antiseparation properties --- of which we have considered intertwined / transitively intertwined here, but there are many more --- becomes central to the story, as they define participants who should decide the same value in a distributed system that tries to achieve consensus.

It is possible to construct a quite extensive theory of semitopological space based on these ideas~\cite{gabbay:semdca,gabbay:semtad}, and to relate these results back to practical systems in ways that are not entirely obvious, including:
\begin{enumerate*}
\item
It can be proved that \emph{any} semitopology partitions into disjoint components whose points are pairwise intertwined within each component.
This goes some way to explaining why blockchains tend to exhibit order~\cite[Theorem~3.5.4, page~32]{gabbay:semdca}.
\item
It can be proved that every semitopology has an actionable \emph{kernel} of participants, such that if they make a decision then all other participants must follow~\cite[Corollary~11.6.11, page~152]{gabbay:semdca}.
This can be read as a distributed systems version of Arrow's theorem~\cite{fey:strpat} (though the proof is different).\footnote{Arrow's theorem proves that dictators exist; the semitopological result, for decentralised systems, is that dictator-\emph{sets} exist.  So the question is: is the dictator-set small relative to the size of the whole semitopology? If so, then this is a measure that the system may be more centralised than desired.}
\item
Semitopological logics can be constructed and used to analyse intertwinedness properties of semitopologies (as documented in~\cite[Chapter~20]{gabbay:semdca}).
\item
In ongoing work, we are using these logics used to formally specify, reason about, and debug consensus protocols.
\end{enumerate*}

\jamiesubsection{Future work}

Semitopologies invite many questions.
We mention just a few:
\begin{enumerate*}
\item
What are the natural semitopological notions of path and homotopy?
\item
What are natural notions of evolution of a semitopology over time?
We ask because in practice, actionable coalitions are not static; they evolve.
Thus, it is natural to consider `deformations' of a semitopology over time.
\item
Following on from the previous point, suppose we are given a semitopology all of whose points are intertwined --- implying, as per Theorem~\ref{thrm.1}, that all points must agree where algorithms succeed.
How close is this semitopology to one such that \emph{not} all of its points are intertwined --- meaning, as discussed, that forking would be possible even where algorithms succeed?
This is a question that would be of interest, for example, to users managing a blockchain to make sure it evolves safely.
\item
Can logics based on semitopologies be used to accelerate development, and increase confidence in, distributed algorithms, by giving new declarative descriptions of consensus algorithms?

This is current work, and so far it has proven useful: most recently we applied a semitopological modal logic to axiomatise the Paxos consensus algorithm~\cite{gabbay2025declarativeapproachspecifyingdistributed}, and in ongoing research we have used a more advanced version of the logic to formally specify and identify errors in a proposed industrial protocol (\emph{Heterogeneous Paxos}~\cite{sheff:hetp,sheff_et_al:LIPIcs.OPODIS.2020.5}).
We are now using the same techniques to help design its replacement.
\end{enumerate*}

\hyphenation{Mathe-ma-ti-sche}
\providecommand{\bysame}{\leavevmode\hbox to3em{\hrulefill}\thinspace}
\providecommand{\MR}{\relax\ifhmode\unskip\space\fi MR }
\providecommand{\MRhref}[2]{%
  \href{http://www.ams.org/mathscinet-getitem?mr=#1}{#2}
}
\providecommand{\href}[2]{#2}


\begin{thebibliography}{SWvRM21}

\bibitem[Dav16]{David2016}
Matthew David, \emph{The {L}egacy of {N}apster}, Networked Music Cultures:
  Contemporary Approaches, Emerging Issues (Rapha{\"e}l Nowak and Andrew
  Whelan, eds.), Palgrave Macmillan UK, London, 2016, DOI:
  10.1057/978-1-137-58290-4\_4, available online at
  \url{https://durham-repository.worktribe.com/output/1642018/} (permalink:
  \url{https://web.archive.org/web/20250421093331/https://durham-repository.worktribe.com/OutputFile/1642039}),
  pp.~49--65.

\bibitem[Eng89]{engelking:gent}
Ryszard Engelking, \emph{General topology}, Sigma Series in Pure Mathematics,
  Heldermann Verlag, 1989.

\bibitem[Fey14]{fey:strpat}
Mark Fey, \emph{{A straightforward proof of Arrow's theorem}}, Economics
  Bulletin \textbf{34} (2014), no.~3, 1792--1797.

\bibitem[Gab24]{gabbay:semdca}
Murdoch~J. Gabbay, \emph{Semitopology: decentralised collaborative action via
  topology, algebra, and logic}, College Publications, August 2024, ISBN
  978-1848904651.

\bibitem[Gab25]{gabbay:semtad}
\bysame, \emph{Semitopology: a topological approach to decentralised
  collaborative action}, The Journal of Logic and Computation (2025),
  \url{https://doi.org/10.1093/logcom/exae050}.

\bibitem[Goo14]{tezos:whitepaper}
L.~M. Goodman, \emph{Tezos -- a self-amending crypto-ledger (white paper)},
  September 2014, \url{https://tezos.com/whitepaper.pdf}.

\bibitem[GZ25]{gabbay2025declarativeapproachspecifyingdistributed}
Murdoch~J. Gabbay and Luca Zanolini, \emph{A declarative approach to specifying
  distributed algorithms using three-valued modal logic}, 2025,
  \url{https://arxiv.org/abs/2502.00892} (submitted for publication).

\bibitem[Rya10]{ryan:hisidf}
Johnny Ryan, \emph{A history of the internet and the digital future}, Reaktion
  Books, 2010, ISBN 978-1861897770.

\bibitem[Rya11]{ars-technica:howabg}
\bysame, \emph{How the atom bomb helped give birth to the internet},
  \url{https://arstechnica.com/tech-policy/2011/02/how-the-atom-bomb-gave-birth-to-the-internet/},
  2 2011, Permalink:
  \url{http://web.archive.org/web/20240622221756/https://arstechnica.com/tech-policy/2011/02/how-the-atom-bomb-gave-birth-to-the-internet/}.

\bibitem[SWvRM20]{sheff:hetp}
Isaac Sheff, Xinwen Wang, Robbert van Renesse, and Andrew~C. Myers,
  \emph{Heterogeneous {P}axos: Technical report}, 2020,
  \url{https://arxiv.org/abs/2011.08253}.

\bibitem[SWvRM21]{sheff_et_al:LIPIcs.OPODIS.2020.5}
Isaac Sheff, Xinwen Wang, Robbert van Renesse, and Andrew~C. Myers,
  \emph{{Heterogeneous Paxos}}, 24th International Conference on Principles of
  Distributed Systems (OPODIS 2020) (Dagstuhl, Germany) (Quentin Bramas, Rotem
  Oshman, and Paolo Romano, eds.), Leibniz International Proceedings in
  Informatics (LIPIcs), vol. 184, Schloss Dagstuhl -- Leibniz-Zentrum f{\"u}r
  Informatik, 2021, pp.~5:1--5:17.

\bibitem[Wil70]{willard:gent}
Stephen Willard, \emph{General topology}, Addison-Wesley, 1970, Reprinted by
  Dover Publications.

\end{thebibliography}
\end{document}